\documentclass[12pt, reqno]{amsart}

\usepackage[margin=.9in, a4paper, foot=.3in]{geometry}

\usepackage{graphicx, color}

\definecolor{darkblue}{rgb}{0,0,.5}
\definecolor{darkred}{rgb}{.5,0,0}
\definecolor{darkgreen}{rgb}{0,0.5,0}

\usepackage{hyperref}
\hypersetup {
    pdfstartview=FitH,
    colorlinks,
    citecolor=darkgreen,
    linkcolor=darkred,
    urlcolor=darkblue
}

\usepackage[noBBpl,slantedGreek]{mathpazo}
\usepackage[mathcal]{euscript}
\usepackage{amssymb}

\usepackage{amssymb}
\allowdisplaybreaks
\numberwithin{equation}{section}
\newtheorem{proposition}{Proposition}[section]

\newcommand {\bil}[2]{(#1 \mid #2)}
\newcommand {\ldbr}{[\![}
\newcommand {\rdbr}{]\!]}

\newcommand {\bbC}{\mathbb C}
\newcommand {\bbE}{\mathbb E}
\newcommand {\bbZ}{\mathbb Z}

\newcommand{\calC}{\mathcal C}

\newcommand {\gothg}{\mathfrak g}
\newcommand {\gothk}{\mathfrak k}

\newcommand {\glmn}{\mathfrak{gl}_{M | N}}
\newcommand {\lslm}{\mathcal L(\mathfrak{sl}_M)}
\newcommand {\lslthree}{\mathcal L(\mathfrak{sl}_3)}
\newcommand {\lsltwo}{\mathcal L(\mathfrak{sl}_2)}
\newcommand {\oo}{{\overline 1}}
\newcommand {\oz}{{\overline 0}}
\newcommand {\uqg}{\mathrm U_q(\mathfrak{g})}
\newcommand {\uqglmn}{\mathrm U_q(\mathfrak{gl}_{M | N})}
\newcommand {\uqglm}{\mathrm U_q(\mathfrak{gl}_M)}
\newcommand {\uqlslm}{\mathrm U_q(\mathcal{L}(\mathfrak{sl}_M))}
\newcommand {\uqlslmn}{\mathrm U_q(\mathcal{L}(\mathfrak{sl}_{M | N}))}

\title[On Poincar\'e--Birkhoff--Witt basis of quantum general linear superalgebra]{On Poincar\'e--Birkhoff--Witt basis\\ of quantum general linear superalgebra}

\author[A. V. Razumov]{Alexander V. Razumov}
\address{Institute for High Energy Physics, NRC ``Kurchatov Institute", 142281 Protvino, Mos\-cow region, Russia}
\email{Alexander.Razumov@ihep.ru}

\begin{document}

\addtolength {\jot}{3pt}

\begin{abstract}
We give a detailed derivation of the commutation relations for the Poincar\'e--Birkhoff--Witt generators of the quantum superalgebra $\mathrm U_q(\mathfrak{gl}_{M|N})$.
\end{abstract}

\maketitle

\tableofcontents

\section{Introduction}

The functional relations are an effective method for investigation of quantum integrable systems. To derive them it is convenient to use the quantum algebraic approach. Previously, what we call quantum algebra was usually called quantum group. In fact, this object is an associative algebra, which in a sense is a deformation of the universal enveloping algebra of a Lie algebra. Nowadays, the term quantum algebra is more commonly used, and we adhere to this terminology. The general notion of a quantum algebra $\uqg$, used in the present paper, was proposed by Drinfeld and Jimbo \cite{Dri85, Jim85} for the case when $\gothg$ is a Kac--Moody algebra with a symmetrizable generalized Cartan matrix.

The derivation of the functional relations based on the quantum algebraic approach was given in the papers \cite{BazLukZam99, BooGoeKluNirRaz13, BooGoeKluNirRaz14a, NirRaz16a, Raz21} for the loop Lie algebra $\gothg = \lsltwo$, in the papers \cite{BazHibKho02, BooGoeKluNirRaz14b, Raz21} for $\gothg = \lslthree$, and in the paper \cite{Raz21a} we gave the derivation for $\gothg = \lslm$ with an arbitrary $M$.\footnote{See also the paper \cite{Koj08}, where some functional relations for $\gothg = \lslm$ were presented without derivation.} The derivation of the functional relations given in the papers \cite{Raz21, Raz21a} is based on the results of the papers \cite{NirRaz17b, BooGoeKluNirRaz16, BooGoeKluNirRaz17b}. In the paper \cite{NirRaz17b}, using the commutation relations for the Poincar\'e--Birkhoff--Witt generators of the quantum algebra $\uqglm$ presented in the paper \cite{Yam89}, we found their action in the Verma $\uqglm$-module. Using some limiting procedure, we found a set of $q$-oscillator modules over the positive Borel subalgebra of $\uqglm$. This modules, via Jimbo's homomorphism were used to construct the corresponding modules over the positive Borel subalgebra of  $\uqlslm$, which are used to construct $Q$-operators.\footnote{For the terminology used for integrability objects, we refer to the papers \cite{BooGoeKluNirRaz14a, Raz21, Raz21a}.} Finally, we derived the corresponding functional relations in the paper \cite{Raz21a}. Here, to analyze the tensor products of the $q$-oscillator modules, we used their $\ell$-weights found in the papers \cite{BooGoeKluNirRaz16, BooGoeKluNirRaz17b}.

By generalizing the defining relations of quantum algebra appropriately, one arrives at quantum algebras associated with the Lie superalgebras \cite{Yam94}. It would interesting to generalize the procedure of constructing the functional relations to the case of quantum superalgebras.\footnote{The first results in this direction were obtained in the paper \cite{BazTsu08}.} It seems that the right choice is to start with the quantum superalgebra $\uqlslmn$. Here the very first step should be derivation of the commutation relations for the Poincar\'e--Birkhoff--Witt generators of the quantum algebra $\uqglmn$. Actually, the commutation relations for this case already were presented in the papers \cite{Zha93, Tsu14, Tsu19} without proof. There is some disagreement between these papers. This fact prompted us to rederive the results of the papers \cite{Zha93, Tsu14, Tsu19}.

The structure of the paper is as follows. In section 2 we remind the necessary facts on the Lie superalgebra $\glmn$.  In section 3 we define the quantum superalgebra $\uqglmn$. The detailed proof of the commutation relation is given in section 4.

We fix the deformation parameter $\hbar$ in such a way that $q = \exp(\hbar)$ is not a root of unity and assume that
\begin{equation*}
q^\nu = \exp (\hbar \nu)
\end{equation*}
for any $\nu \in \bbC$. We define $q$-numbers by the equation
\begin{equation*}
[\nu]_q = \frac{q^\nu - q^{- \nu}}{q - q^{-1}}, \qquad \nu \in \bbC
\end{equation*}

\section{\texorpdfstring{Lie superalgebra $\glmn$}{Lie superalgebra glM|N}} \label{s:2}

We fix two positive integers $M$ and $N$ such that $M, N \ge 1$ and $M \ne N$, and denote by $\bbC_{M | N}$ the superspace\footnote{See appendix A of the paper \cite{Raz22} for a minimal set of definitions and notation.} formed by $(M + N)$-tuples of complex numbers with the following grading. An element of $\bbC_{M | N}$ is even if its last $N$ components are zero, and odd if its first $M$ components are zero. For simplicity, we denote the Lie superalgebra $\mathfrak{gl}(\bbC_{M | N})$ as $\glmn$. We denote by $v_i$, $i = 1, \ldots, M + N$, the elements of the standard basis of $\bbC_{M | N}$. By definition,
\begin{equation*}
[v_i] = \oz, \quad i = 1, \ldots, M, \qquad [v_i] = \oo, \quad i = M + 1, \ldots, N.
\end{equation*}
It is convenient to use the notation
\begin{equation*}
[i] = [v_i], \qquad i = 1, \ldots, M + N.
\end{equation*}
The elements $\bbE_{i j} \in \glmn$, $i, j = 1, \ldots, M + N$, defined by the equation
\begin{equation*}
\bbE_{i j} v_k = v_i \delta_{j k},
\end{equation*}
form a basis of the Lie superalgebra $\glmn$. It is clear that the matrices of $\bbE_{i j}$ with respect to the standard basis of $\bbC_{M | N}$ are the usual matrix units, and we have
\begin{equation*}
\bbE_{i j} \bbE_{k l} = \delta_{j k} \bbE_{i l}.
\end{equation*}
It is also evident that
\begin{equation*}
[\bbE_{i j}] = [i] + [j].
\end{equation*}

As the Cartan subalgebra $\gothk$ of the Lie superalgebra $\glmn$ we take the subalgebra span\-ned by the elements $K_i = \bbE_{i i}$, $i = 1, \ldots, M + N$, which form its basis. Denote by $(\Xi_i)_{i = 1, \ldots, M + N}$ the dual basis of the space $\gothk^*$. For $X = \sum_{i = 1}^{M + N} c_i K_i \in \gothk$ we have
\begin{equation*}
[X, \, \bbE_{i j}] = (c_i - c_j) \, \bbE_{i j} = \langle\Xi_i -\Xi_j, \, X \rangle \, \bbE_{i j}.
\end{equation*}
Hence, $\bbE_{i j}$, $i \ne j$, is a root vector corresponding to the root $\Xi_i -\Xi_j$ and the root system of $\glmn$ is the set
\begin{equation*}
\Delta = \{\Xi_i -\Xi_j \mid i, j = 1, \ldots, M + N, \ i \ne j\}.
\end{equation*}
We choose as the system of simple roots the set
\begin{equation*}
\Pi = \{\Xi_i -\Xi_{i + 1} \mid i = 1, \ldots, M + N - 1\},
\end{equation*}
then the system of positive roots corresponding to $\Pi$ is
\begin{equation*}
\Delta_+ = \{\alpha_{i j} = \Xi_i -\Xi_j \mid 1 \le i < j \le M + N\}.
\end{equation*}
Certainly, the corresponding system of negative roots is $\Delta_- = -\Delta_+$. Denoting
\begin{equation*}
\alpha_i =  \alpha_{i, \, i + 1} =\Xi_i -\Xi_{i + 1}, \qquad i = 1, \ldots, M + N - 1,
\end{equation*}
we obtain
\begin{equation*}
\alpha_{i j} =  \sum_{k = 1}^{j - 1} \alpha_k, \qquad 1 \le i < j \le M + N.
\end{equation*}
We define a strict partial order $\prec$ on $\gothk^*$ as follows. Given $\alpha, \beta \in \gothk^*$, we assume that $\beta \prec \alpha$ if and only if $\alpha - \beta$ is the sum of positive roots.

Define a nondegenerate symmetric bilinear form $\bil {\cdot}{\cdot}$ on $\gothk^*$ by the equation
\begin{equation*}
\bil {\Xi_i} {\Xi_j} = (-1)^{[i]} \delta_{i j} = d_i \delta_{i j},
\end{equation*}
where
\begin{equation*}
d_i = (-1)^{[i]}.
\end{equation*}
We see that
\begin{equation*}
\bil {\alpha_{i j}} {\alpha_{m n}} = d_i \delta_{i m} - d_j \delta_{j m} - d_i \delta_{i n} + d_j \delta_{j n},
\end{equation*}
Below we often use the relations
\begin{gather}
\bil {\alpha_{i j}}{\alpha_{j n}} = -d_j, \qquad \bil {\alpha_{i j}}{\alpha_{m i}} = -d_i, \label{bfa} \\*
\bil {\alpha_{i j}}{\alpha_{i n}} = d_i, \quad j \ne n, \qquad \bil {\alpha_{i j}}{\alpha_{i j}} = d_i + d_j, \qquad \bil {\alpha_{i j}}{\alpha_{m j}} = d_j, \quad i \ne m. \label{bfb}
\end{gather}
In fact, these are all nonzero cases.

\section{\texorpdfstring{Quantum superalgebra $\uqglmn$}{Quantum superalgebra Uq(glm|n)}} \label{s:3}

We define the quantum superalgebra $\uqglmn$ as a unital associative $\bbC$-superalgebra generated by the elements\footnote{We use capital letters to distinguish between generators of the quantum superalgebra $\uqglmn$ and the quantum superalgebra $\uqlslmn$.}
\begin{equation*}
E_i, \quad F_i, \quad i = 1, \ldots, M + N - 1, \qquad q^X, \quad X \in \gothk,
\end{equation*}
which obey the corresponding defining relations. The $\bbZ_2$-grading of the quantum superalgebra $\uqglmn$ is defined on generators as
\begin{equation*}
[q^X] = \oz, \qquad [E_i] = [F_i] = \left\{ \begin{array}{ll}
\oz, & i \ne M, \\[.5em]
\oo, & i = M.
\end{array} \right.
\end{equation*}

Before giving the explicit form of the defining relations, introduce the notion of the $q$-supercommutator. The abelian group
\begin{equation*}
Q = \bigoplus_{i = 1}^{M + N - 1} \bbZ \, \alpha_i.
\end{equation*}
is called the root lattice of the Lie superalgebra $\glmn$. Assuming that
\begin{equation*}
q^X \in \uqglmn_0, \qquad E_i \in \uqglmn_{\alpha_i}, \qquad F_i \in \uqglmn_{-\alpha_i},
\end{equation*}
we endow $\uqglmn$ with a $Q$-grading. Now, for any elements $X \in \uqglmn_\alpha$ and $Y \in \uqglmn_\beta$ we define the $q$-supercommutator by the equation
\begin{align*}
& \ldbr X, \, Y \rdbr =  X Y - (-1)^{[X] [Y]} q^{- \bil \alpha \beta} Y X = X Y - (-1)^{[\alpha] [\beta]} q^{- \bil \alpha \beta} Y X 
\intertext{if $\alpha, \beta \succ 0$, by the equation}
& \ldbr X, \, Y \rdbr =  X Y - (-1)^{[X] [Y]} q^{\bil \alpha \beta} Y X = X Y - (-1)^{[\alpha] [\beta]} q^{\bil \alpha \beta} X Y 
\intertext{if $\alpha, \beta \prec 0$, and by the equation}
& \ldbr X, \, Y \rdbr =  X Y - (-1)^{[X] [Y]} Y X = X Y - (-1)^{[\alpha] [\beta]} Y X
\end{align*}
if $\alpha \prec 0$ and $\beta \succ 0$, or $\alpha \succ 0$ and $\beta \prec 0$.

The defining relations of the quantum superalgebra $\uqglmn$ have the form \cite{Yam94}
\begin{gather}
q^0 = 1, \qquad q^{X_1} q^{X_2} = q^{X_1 + X_2}, \label{glra} \\
q^X E_i q^{-X} = q^{\langle \alpha_i, \, X \rangle} E_i, \qquad q^X F_i q^{-X} = q^{-\langle \alpha_i, \, X \rangle} F_i, \label{glrb} \\
\ldbr E_i, \, F_j \rdbr = \delta_{i j} \frac{q^{d_i K_i - d_{i + 1} K_{i + 1}} -q^{- d_i K_i + d_{i + 1} K_{i + 1}}}{q_i^{} - q_i^{- 1}},  \label{glrc}
\end{gather}
where $i, j = 1, \ldots, M + N - 1$. Here and below we use the notation
\begin{equation*}
q_i =  q^{d_i} = q^{(-1)^{[i]}}.
\end{equation*}
It is useful to have in mind that
\begin{equation*}
[2]_{q^i} = q_i^{} + q_i^{-1} = q + q^{-1} = [2]_q
\end{equation*}
and
\begin{equation}
(q^{}_i - q^{-1}_i) = d_i (q - q^{-1}) = (-1)^{[i]} (q - q^{-1}). \label{kappaqi}
\end{equation}
There are also the following Serre relations
\begin{gather}
\ldbr E_i, \, E_j \rdbr = 0, \qquad  \ldbr F_j, \, F_i \rdbr = 0, \qquad \bil {\alpha_i} {\alpha_j} = 0,  \label{glsra} \\
\ldbr \ldbr E_{i - 1}, \, E_i \rdbr, \, E_i \rdbr = 0, \qquad \ldbr F_i, \ldbr \, F_i , \, F_{i - 1} \rdbr \rdbr = 0, \qquad \bil {\alpha_i} {\alpha_i} \ne 0, \label{glsrcm} \\
\ldbr E_i, \, \ldbr E_i, \, E_{i + 1} \rdbr \rdbr = 0, \qquad \ldbr \ldbr F_{i + 1}, \, F_i \rdbr , \, F_i \rdbr = 0, \qquad \bil {\alpha_i} {\alpha_i} \ne 0, \label{glsrcp} \\
\hspace{2em} \ldbr \ldbr \ldbr E_{M - 1}, \, E_M \rdbr, \, E_{M + 1} \rdbr, \, E_M \rdbr = 0, \qquad \ldbr F_M, \, \ldbr F_{M + 1}, \, \ldbr F_M, \, F_{M - 1} \rdbr \rdbr \rdbr = 0.  \label{glsre}
\end{gather}

Let us rewrite the defining relations (\ref{glsra})--(\ref{glsrcp}) in a more familiar form. The relations (\ref{glsra}) are equivalent to the equations
\begin{gather}
E_i E_j = E_j E_i, \qquad F_j F_i = F_i F_j. \qquad |i - j| > 1, \label{imjgo} \\
E_M^2 = 0, \qquad F_M^2 = 0,
\end{gather}
and the relations (\ref{glsrcm})--(\ref{glsrcp}) are equivalent to
\begin{align}
& E^{}_{i - 1} E_i^2  - [2]_q E^{}_i E^{}_{i - 1} E^{}_i + E_i^2 E^{}_{i - 1} = 0, &&
F_i^2 F^{}_{i - 1} - [2]_q F^{}_i F^{}_{i - 1} F^{}_i + F_{i - 1} F_i^2 = 0, \label{eesm} \\*
& E_i^2 E^{}_{i + 1} - [2]_q E^{}_i E^{}_{i + 1} E^{}_i + E^{}_{i + 1} E_i^2 = 0, &&
F^{}_{i + 1} F_i^2 - [2]_q F^{}_i F^{}_{i + 1} F^{}_i +  F_i^2 F^{}_{i + 1} = 0.\label{esep}
\end{align}
where $i \ne M$.

\section{\texorpdfstring{Poincar\'e--Birkhoff--Witt basis of $\uqglmn$}{Poincare-Birkhoff-Witt basis of Uq(gl(M|N))}}

An element $a$ of $\uqglmn$ is called a root vector corresponding to a root $\gamma$ of $\glmn$ if $a \in \uqglmn_\gamma$. In particular, $E_i$ and $F_i$ are root vectors corresponding to the roots $\alpha_i$ and $- \alpha_i$. It is possible to  construct linearly independent root vectors corresponding to all roots of $\glmn$. To this end, being inspired by M.~Jimbo \cite{Jim86a}, we introduce elements $E_{ij}$ and $F_{ij}$, $1 \le i < j \le M + N$, with the help of the relations
\begin{gather}
E_{i, \, i + 1} = E_i, \qquad F_{i, \, i + 1} = F_i, \label{41a} \\
E_{i, \, j + 1} = \ldbr  E_{i j}, \, E_{j, \, j + 1} \rdbr, \qquad F_{i, \, j + 1} = \ldbr  F_{j, \,  j + 1}, \, F_{i j} \rdbr, \qquad j > i. \label{41b}
\end{gather} 
Explicitly, the last two equations look as
\begin{equation*}
E_{i, \, j + 1} = E_{i j} \, E_{j, \, j + 1} - q_j \, E_{j, \, j + 1} \, E_{i j}, \qquad F_{i, j + 1 } = F_{j, \, j + 1} \, F_{i j} - q_j^{-1} F_{i j} \, F_{j, \, j + 1}.
\end{equation*}
Note that we have
\begin{equation*}
[E_{i j}] = [i] + [j],
\end{equation*}
in particular,
\begin{equation*}
[E_i] = [F_i] = [i] + [i + 1].
\end{equation*}
We also see that
\begin{equation*}
[E_{i j}] = \oz
\end{equation*}
if and only if $j < M$ or $i > M$,
\begin{equation*}
[E_{i j}] = \oo
\end{equation*}
if and only if $i \le M < j$.

It is clear that the vectors $E_{i j}$ and $F_{i j}$ correspond to the roots $\alpha_{i j}$ and $- \alpha_{i j}$ respectively. These vectors are linearly independent, and together with the elements $q^X$, $X \in \gothk$, are called Cartan--Weyl generators of $\uqglmn$. It appears that the ordered monomials constructed from the Cartan--Weyl generators form a Poincar\'e--Birkhoff--Witt basis of $\uqglmn$. In this paper we choose the following total order for monomials. First, we endow the set of the pairs $(i, \, j)$, where $1 \le i < j \le M + N$, with the lexicographical order. It means that $(i, \, j) \prec (m, \, n)$ if $i < m$, or if $i = m$ and $j < n$.\footnote{Note that if we define an ordering of positive roots so that $\alpha_{i j} \prec \alpha_{m n}$ if $(i, \, j) \prec (m, \, n)$ we get a normal ordering in the sense of \cite{LezSav74, AshSmiTol79}, see also \cite{Tol89}.} Now we say that a monomial is ordered if it has the form
\begin{equation}
F_{i_1 j_1} \cdots F_{i_r j_r} \, q^X \, E_{m_1 n_1} \ldots E_{m_s n_s}, \label{41c}
\end{equation}
where $(i_1, \, j_1) \preccurlyeq \cdots \preccurlyeq (i_r, \, j_r)$, $(m_1, \, n_1) \preccurlyeq \cdots \preccurlyeq (m_s, \, n_s)$ and $X$ is an arbitrary element of~$\gothk$. In the present paper we only show that any monomial can be written as a finite sum of monomials of the form (\ref{41c}). To prove that they form a basis of $\uqglmn$ one can use arguments similar to those used in the paper \cite{Yam89} for the the case of the quantum algebra $\uqglm$.

We present the relations necessary for ordering as a sequence of propositions. First consider the ordering of $q^X$ with $E_{i j}$ and $F_{i j}$.

\begin{proposition} \label{p:4.1}
For any $1 \le j < n \le M + N$ and $i = 1, \ldots, M + N$, we have
\begin{equation}
q^{\nu K_i} E_{j n} \, q^{- \nu K_i} = (q^\nu \delta_{i j} + q^{-\nu} \delta_{i n}) E_{j n}, \qquad q^{\nu K_i} F_{j n} \, q^{- \nu K_i} = (q^{-\nu} \delta_{i j} + q^\nu \delta_{i k}) E_{j n}. \label{e41a}
\end{equation}
\end{proposition}

\begin{proof}
It is evident that
\begin{equation*}
\langle \alpha_j, \, K_i \rangle = \delta_{i j} - \delta_{i, \, j + 1},
\end{equation*}
and it follows from the defining relation (\ref{glrb}) that
\begin{equation*}
q_i^{\nu K_i} E_j^{} \, q_i^{- \nu K_i} = q^{\nu (\delta_{i j} - \delta_{i, \, j + 1})} E_j^{}, \qquad q_i^{\nu K_i} F_j^{} \, q_i^{- \nu K_i} = q^{-\nu (\delta_{i j} - \delta_{i, \, j + 1})} F_j^{}.
\end{equation*}
which follow from (\ref{glrb}).

\begin{equation*}
q^{\nu K_i} E_{j n} q^{-\nu K_i} = q^{\nu \sum_{m = j}^{n - 1} (\delta_{i m} - \delta_{i, \, m + 1})} E_{j n} = q^{\nu (\delta_{i j} - \delta_{i n})} E_{j n} = (q^\nu \delta_{i j} + q^{-\nu} \delta_{i n}) E_{j n}
\end{equation*}
Thus, the first equation of (\ref{e41a}) is true. The proof of the second equations is similar.
\end{proof}

Now we consider the ordering of the root vectors $E_{i j}$, $1 \le i < j \le M + N$,  and $F_{i j}$, $1 \le i < j \le M + N$. We divide the set of pairs $((i, \, j), \, (m, \, n))$, where $1 \le i < j \le M + N$, $1 \le m < n \le M + N$ and $(i, \, j) \prec (m, \, n)$, into six branches $\calC_a$, $a = \mathrm I, \ldots, \mathrm{VI}$. The conditions defining the branches are given in table \ref{t:1}.
\begin{table}[h]
\begin{equation*}
\begin{array}{llcc}
\hline \hline \\[-.8em]
& & ([i] + [j])([m] + [n]) & \bil{\alpha_{i j}}{\alpha_{m n}} \\[.2em]
\hline \\[-.8em]
\calC_{\mathrm{I}} & i = m < j < n & [i] + [j] & (-1)^{[i]} \\[.2em]
\calC_{\mathrm{II}} & i < m < n < j & [m] + [n] & 0 \\[.2em]
\calC_{\mathrm{III}} & i < m < j = n & [m] + [j] & (-1)^{[j]} \\[.2em]
\calC_{\mathrm{IV}} & i < m < j  < n & [m] + [j] & 0 \\[.2em]
\calC_{\mathrm {V}} & i < j = m < n & \oz & -(-1)^{[j]} \\[.2em]
\calC_{\mathrm{VI}} & i < j < m < n & \oz & 0 \\[.2em]
\hline
\end{array}
\end{equation*}
\caption{}  \label{t:1}
\end{table}
In the same table we put the information necessary to construct the corresponding $q$-supercommutators. To fill table~\ref{t:1} it is sufficient to use the relations
\begin{equation*}
a^2 = a, \qquad a + a = \oz, \qquad a \in \bbZ_2,
\end{equation*}
and equations (\ref{bfa}) and (\ref{bfb}).

%\subsection*{$\calC_{\mathrm{VI}}$}

\begin{proposition} \label{p:4.2}
For any $((i, \, j), \, (m, \, n)) \in \calC_{\mathrm{VI}}$ one has
\begin{equation}
\ldbr E_{i j}, \, E_{m n} \rdbr = E_{i j} E_{m n} - E_{m n} E_{i j} = 0, \qquad  \ldbr F_{m n}, \, F_{i j} \rdbr = F_{m n} F_{i j} - F_{i j} F_{m n} = 0. \label{e42a}
\end{equation}
\end{proposition}

\begin{proof}
The statement of the proposition is a direct consequence of the Serre relations (\ref{imjgo}).
\end{proof}

%\subsection*{$\calC_{\mathrm{V}}$}

\begin{proposition} \label{p:4.3}
For any $((i, \, j), \, (m, \, n)) \in \calC_{\mathrm{V}}$ one has
\begin{equation*}
\ldbr E_{i j}, \, E_{m n} \rdbr = E_{i j} E_{j n} - q_j E_{j n} E_{i j} = E_{i n}, \qquad \ldbr F_{m n}, \, F_{i j} \rdbr = F_{j n} F_{i j} - q_j^{-1} F_{i j} F_{j n} = F_{i n}.
\end{equation*}
\end{proposition}

\begin{proof}
The proposition can be proved by induction over $n$. For $n = j + 1$ we have just the definition (\ref{41b}). Assume that the statement of the proposition is valid for some given $n > j$, then we have
\begin{equation*}
E_{i j} E_{j n} - q_j E_{j n} E_{i j} = E_{i n}.
\end{equation*}
Using this equation and proposition \ref{p:4.2}, we get
\begin{align*}
\ldbr & E_{i j}, \, E_{j, \, n + 1} \rdbr = E_{i j} E_{j, \, n + 1} - q_j E_{j, \, n + 1} E_{i j} = E_{i j}(E_{j n} E_{n, \, n + 1} - q_n E_{n, \, n + 1} E_{j n}) \\
& {} - q_j (E_{j n} E_{n, \, n + 1} - q_n E_{n, \, n + 1} E_{j n}) E_{i j} = (E_{i j} E_{j n} - q_j E_{j n} E_{i j}) E_{n, \, n + 1} \\
& \hspace{10em} - q_n E_{n, \, n + 1} (E_{i j} E_{j n} - q_j E_{j n} E_{i j}) = \ldbr E_{i n}, \, E_{n, n + 1} \rdbr = E_{i, \, n + 1}.
\end{align*}
Thus, the first equation of the proposition is true. The second one can be proved in the same way.
\end{proof}

%\subsection*{$\calC_{\mathrm{II}}$}

\begin{proposition} \label{p:4.4}
For any $((i, \, j), \, (m, \, n)) \in \calC_{\mathrm{II}}$ one has
\begin{align*}
& E_{i j}, \, E_{m n} \rdbr = E_{i j} E_{m n} - (-1)^{[m] + [n]} E_{m n} E_{i j} = 0, \\
& \ldbr F_{m n}, \, F_{i j} \rdbr = F_{m n} F_{i j} - (-1)^{[m] + [n]} F_{i j} F_{m n} = 0.
\end{align*}
\end{proposition}

\begin{proof}
Let us first prove that
\begin{equation}
\ldbr E_{m - 1, \, m + 2}, \, E_{m, \, m + 1} \rdbr = 0 \label{e44a}
\end{equation}
for any $2 \le m \le M + N - 2$. It is easy to see that for  $m = M$, equation (\ref{e44a}) is just the Serre relation (\ref{glsre}). If $m \ne M$, we have\footnote{It is clear that either $m < M$ or $m > M$, so that $q_m = q_{m + 1}$.}
\begin{align*}
& \ldbr E_{m - 1, \, m + 2}, \, E_{m, \, m + 1} \rdbr = \ldbr \ldbr \ldbr E_{m - 1, \, m}, \, E_{m, \, m + 1} \rdbr, \, E_{m + 1, \, m + 2} \rdbr, \, E_{m, \, m + 1} \rdbr \\*
& \hspace{4em} {} = E_{m - 1} E_m E_{m + 1} E_m - q_m E_{m - 1} E_{m + 1} E_m^2 - q_m E_m E_{m + 1} E_{m - 1} E_m  \\*
& \hspace{4em} {} + q_m^2 E_{m + 1} E_m E_{m -1} E_m - E_m E_{m - 1} E_m E_{m + 1} + q_m E_m E_{m - 1} E_{m + 1} E_m \\*
& \hspace{15em} {}  + q_m E_m^2 E_{m + 1} E_{m - 1} - q_m^2 E_m E_{m + 1} E_m E_{m - 1}.
\end{align*}
Using the first equations of (\ref{eesm}) and (\ref{esep}), we obtain
\begin{align*}
& \ldbr E_{m - 1, \, m + 2}, \, E_{m, \, m + 1} \rdbr \\
& \hspace{1em} {} = [2]_q^{-1} E_{m + 1} (E_m^2 E_{m - 1} + E_{m + 1} E_m^2)- q_m E_{m - 1} E_{m + 1} E_m^2 - q_m E_m E_{m - 1} E_{m + 1} E_m  \\
& \hspace{1em} {} + [2]_q^{-1} q_m^2 E_{m + 1} (E_m^2 E_{m - 1} + E_{m - 1} E_m^2) - [2]_q^{-1} (E_m^2 E_{m - 1} + E_{m - 1} E_m^2) E_{m + 1} \\
& \hspace{1em} {} + q_m E_m E_{m - 1} E_{m + 1} E_m + q_m E_m^2 E_{m + 1} E_{m - 1} - [2]_q^{-1} q_m^2 (E_m^2 E_{m + 1} + E_{m + 1} E_m^2) E_{m - 1}.
\end{align*}
The Serre relations (\ref{glsra}) give $E_{m - 1} E_{m + 1} = E_{m + 1} E_{m - 1}$, and we see that equation (\ref{e44a}) is true for any admissible value of $m$.

Assume that
\begin{equation}
\ldbr E_{i j}, E_{m, \, m + 1} \rdbr = 0 \label{e44b}
\end{equation}
for some $2 \le i < m < j - 1 \le M + N - 1$. We have
\begin{multline*}
\ldbr E_{i - 1, \, j}, E_{m, \, m + 1} \rdbr = \ldbr E_{i - 1, i} E_{i j} - q_i E_{i j} E_{i - 1, \, i}, E_{m, \, m + 1} \rdbr \\
= E_{i - 1, \, i} \ldbr E_{i j}, \, E_{m, \, m + 1} \rdbr - q_i \ldbr E_{i j}, \, E_{m, \, m + 1} \rdbr E_{i - 1, \, i} = 0.
\end{multline*}
If equation (\ref{e44b}) is valid for some $1 \le i < m < j - 1 \le M + N - 2$, then
\begin{multline*}
\ldbr E_{i, \, j + 1}, E_{m, \, m + 1} \rdbr = \ldbr E_{i j} E_{j, \, j + 1} - q_j E_{j, \, j + 1} E_{i j}, \, E_{m, \, m + 1} \rdbr \\
= \ldbr E_{i j}, \, E_{m, \, m + 1} \rdbr E_{j, \, j + 1} - q_j E_{j, \, j + 1} \ldbr E_{i j} , \, E_{m, \, m + 1} \rdbr = 0.
\end{multline*}
Thus, equation (\ref{e44b}) is valid for any admissible $i$, $j$ and $m$.

Finally, assume that the equation
\begin{equation*}
\ldbr E_{i j}, E_{m n} \rdbr = E_{i j} E_{m n} - (-1)^{[m] + [n]} E_{m n} E_{i j} = 0
\end{equation*}
is valid for some $1 \le i < m < n < j - 1 \le M + N - 1$, then we have
\begin{multline*}
\ldbr E_{i j}, \, E_{m, \, n + 1} \rdbr = \ldbr E_{i j}, \, E_{m n} E_{n, \, n + 1} - q_n E_{n, \, n + 1} E_{m n} \rdbr \\*
{} = \ldbr E_{i j}, \, E_{m n} \rdbr E_{n, \, n + 1} - (-1)^{[n] + [n + 1]} q_n E_{n, \, n + 1} \ldbr E_{i j}, \, E_{m n} \rdbr = 0. 
\end{multline*}

Now, it is clear that the first equation of the proposition is valid. The second equation of the proposition can be proved in a similar way. 
\end{proof}

%\subsection*{$\calC_{\mathrm{I}} \cup \calC_{\mathrm{III}}$}

\begin{proposition} \label{p:4.5}
For any $((i, \, j), \, (m, \, n)) \in \calC_{\mathrm{I}}$ one has
\begin{align}
& \ldbr E_{i j}, \, E_{m n} \rdbr = E_{i j} E_{i n} - (-1)^{[i] + [j]} q_i^{-1} E_{i n} E_{i j} = 0,  \label{e35a} \\
& \ldbr F_{m n}, \, F_{i j} \rdbr  = F_{i n} F_{i j} - (-1)^{[i] + [j]} q_i F_{i j} F_{i n} = 0. \label{e35b}
\end{align}
For any $((i, \, j), \, (m, \, n)) \in \calC_{\mathrm{III}}$ one has
\begin{align}
& \ldbr E_{i j}, \, E_{m n} \rdbr = E_{i j} E_{m j} - (-1)^{[m] + [j]} q_j^{-1} E_{m j} E_{i j} = 0, \label{e35c} \\*
& \ldbr F_{m n}, \, F_{i j} \rdbr = F_{m j} F_{i j} - (-1)^{[m] + [j]} q_j F_{i j} F_{m j} = 0. \label{e35d}
\end{align}
\end{proposition}

\begin{proof}
Let us consider the case $((i, \, j), \, (m, \, n)) \in \calC_{\mathrm{I}}$ and prove equation (\ref{e35a}). First we demonstrate that
\begin{equation}
\ldbr E_{i, \, i + 1}, \, E_{i, \, i + 2} \rdbr = 0 \label{e35e}
\end{equation}
for any $1 \le i \le M + N - 2$. We have
\begin{equation*}
\ldbr E_{i, \, i + 1}, \, E_{i, \, i + 2} \rdbr = \ldbr E_{i, \, i + 1}, \, \ldbr E_{i, \, i + 1}, \, E_{i + 1, \, i + 2} \rdbr \rdbr =  \ldbr E_i, \, \ldbr E_i, \, E_{i + 1} \rdbr \rdbr.
\end{equation*}
Hence, for $i \ne M$, the equation (\ref{e35e}) is equivalent to the first of the Serre relations (\ref{glsrcp}). For $i = M$ we obtain 
\begin{multline*}
\ldbr E_{i, \, i + 1}, \, E_{i, \, i + 2} \rdbr = \ldbr E_{M, \, M + 1}, \, \ldbr E_{M, \, M + 1}, \, E_{M + 1, \, M + 2} \rdbr \rdbr \\
= \ldbr E_M, \, E_M E_{M + 1} - q_{M + 1} E_{M + 1} E_M \rdbr = -(q_{M + 1}^{} - q_M^{-1}) E_M E_{M + 1} E_M = 0. 
\end{multline*}
Thus, equation (\ref{e35e}) is valid for any $1 \le i \le M + N - 2$. Assume that
\begin{equation}
\ldbr E_{i, \, i + 1}, \, E_{i n} \rdbr = E_{i, \, i + 1} E_{i n} - (-1)^{[i] + [i + 1]} q_i^{-1} E_{i n} E_{i, \, i + 1} = 0 \label{e35f}
\end{equation}
for some $1 \le i < n - 1 \le M + N - 2$. Using equation (\ref{e35f}), we obtain
\begin{multline*}
\ldbr E_{i, \, i + 1}, \, E_{i, \, n + 1} \rdbr = \ldbr E_{i, \, i + 1}, \, E_{i n} E_{n, \, n + 1} - q_n E_{n, \, n + 1} E_{i n} \rdbr = E_{i, \, i + 1} E_{i n} E_{n, \, n + 1} \\
{} - q_n E_{i, \, i + 1} E_{n, \, n + 1} E_{i n} - (-1)^{[i] + [i + 1]} q_i^{-1}  (E_{i n} E_{n, \, n + 1} E_{i, \, i + 1} - q_n E_{n, \, n + 1} E_{i n} E_{i, \, i + 1}) = 0.
\end{multline*}
 we obtain that $\ldbr E_{i, \, i + 1}, \, E_{i, \, n + 1} \rdbr = 0$. It follows that
\begin{equation*}
\ldbr E_{i, \, i + 1}, \, E_{i n} \rdbr = 0
\end{equation*}
for any $1 \le i < n - 1 \le M + N - 1$. Now, assume that
\begin{equation*}
\ldbr E_{i j}, E_{i n} \rdbr = E_{i j} E_{i n} - (-1)^{[i] + [j]} q_i^{-1} E_{i n} E_{i j} = 0
\end{equation*}
for some $1 \le i < j + 1 < n \le M + N$. Using proposition \ref{p:4.4}, we get
\begin{multline*}
\ldbr E_{i, \, j + 1}, \, E_{i n} \rdbr = \ldbr E_{i j} E_{j, \, j + 1} - q_j E_{j, \, j + 1} E_{i j}, \, E_{i n} \rdbr \\
= (-1)^{[j] + [j + 1]} \ldbr E_{i j}, \, E_{i n} \rdbr E_{j, \, j + 1} - q_j  E_{j, \, j + 1} \ldbr E_{i j}, \, E_{i n} \rdbr = 0.
\end{multline*}
Thus, for $((i, \, j), \, (m, \, n)) \in \calC_{\mathrm{I}}$, equation (\ref{e35a}) is true. Equation (\ref{e35b}) can be proved in the same way. In the case when $((i, \, j), \, (m, \, n)) \in \calC_{\mathrm{III}}$, one can prove equations (\ref{e35c}) and (\ref{e35d}) in a similar way. 
\end{proof}

It follows from the above proposition that if $((i, \, j), \, (m, \, n)) \in \calC_{\mathrm{I}}$, then
\begin{equation}
\ldbr  E_{i j}, \, \ldbr E_{i j}, \, E_{j n} \rdbr \rdbr = 0, \qquad \ldbr \ldbr F_{j n}, \, F_{i j} \rdbr, \, F_{i j} \rdbr = 0, \label{eeei}
\end{equation}
and if $((i, \, j), \, (m, \, n)) \in \calC_{\mathrm{III}}$, then
\begin{equation}
\ldbr \ldbr E_{i m}, \, E_{m n} \rdbr, \, E_{m n} \rdbr = 0, \qquad \ldbr F_{m n}, \, \ldbr F_{m n}, \, F_{i m} \rdbr \rdbr = 0. \label{eeeii}
\end{equation}
These relations are a generalization of the Serre relations (\ref{glsrcm}) and (\ref{glsrcp}).

Note that the quantum supergroup $\uqglmn$ has two natural subgroups isomorphic to $\mathrm U_q(\mathfrak{gl}_{M})$ and $\mathrm U_q(\mathfrak{gl}_{N})$. The former is generated by $E_i$, $F_i$, $i = 1, \ldots, M - 1$, and $q^X$, where $X$ belongs to the linear span of the elements $K_i$, $i = 1, \ldots M$, and the latter is generated by $E_i$, $F_i$, $i = M + 1, \ldots, M + N - 1$, and $q^X$, where $X$ belongs to the linear span of the elements $K_i$, $i = M + 1, \ldots M + N$. It is clear that $[i] + [j] = \oz$ iff $E_{i j}$ belongs to one of these two subgroups. Each of them has no zero divisors, see the paper \cite{Yam89}. Hence, for any element $E_{i j}$ belonging to them one has $E_{i j}^2 \ne 0$. In other words, if $[i] + [j] = \oo$ then $E_{i j}^2 \ne 0$.

%\subsection*{$\calC$}

\begin{proposition} \label{p:4.6}
For all $1 \le i < j \le M + N$ such that $[i] + [j] = \oo$ one has
\begin{equation}
\frac{1}{2} \, \ldbr E_{i j}, \, E_{i j} \rdbr = E_{i j}^2 = 0. \label{eijs}
\end{equation}
\end{proposition}

\begin{proof}
In fact, we should demonstrate that if $i \le M < j$, then
\begin{equation}
E_{i j}^2 = 0. \label{eijse}
\end{equation}
First, we show that
\begin{equation}
E_{M j}^2 = 0 \label{emmpk}
\end{equation}
for all $j > M$. It is certainly the case, at least for $j = M + 1$. Using the fact that $q_j = q^{-1}$ for any $j > M$, we obtain
\begin{multline}
E^2_{M, \, j + 1} = (E_{M j} E_{j, \, j + 1} - q_j E_{j, \, j + 1} E_{M j})^2 \\
{} = E_{M j} E_{j, \, j + 1} E_{M j} E_{j, \, j + 1} - q^{-1} E_{M j} E^2_{j, \, j + 1} E_{M j} + q^{-2} E_{j, \, j + 1} E_{M j} E_{j, \, j + 1} E_{M j}. \label{esm}
\end{multline}
It follows from the first relation of (\ref{eeeii}) that
\begin{equation*}
\ldbr \ldbr E_{M j}, \, E_{j, \, j + 1} \rdbr, \, E_{j, \, j + 1} \rdbr = 0,
\end{equation*}
or, in a more explicit form,
\begin{equation*}
E_{M j}^{} E^2_{j, \, j + 1} - [2]_q E_{j, \, j + 1}^{} E_{M j}^{} E_{j, \, j + 1}^{} + E^2_{j, \, j + 1} E_{M j}^{}  = 0,
\end{equation*}
Multiplying this equation from the left and from the right by $E_{M j}$, we obtain
\begin{align}
- [2]_q E_{M j}^{} E_{j, \, j + 1}^{} E_{M j}^{} E_{j, \, j + 1}^{} + E_{M j}^{} E^2_{j, \, j + 1} E_{M j}^{}  & = 0, \\
E_{M j}^{} E^2_{j, \, j + 1} E_{M j}^{} - [2]_q E_{j, \, j + 1}^{} E_{M j}^{} E_{j, \, j + 1}^{} E_{M j}^{} & = 0. \label{esee}
\end{align}
It follows that
\begin{equation*}
E_{M j}  E_{j, \, j + 1} E_{M j} E_{j, \, j + 1} = E_{j, \, j + 1} E_{M j} E_{j, \, j + 1} E_{M j}.
\end{equation*}
Using this equation in (\ref{esm}), we get
\begin{equation*}
E^2_{M, \, j + 1} = - q^{-1} ({} - [2]_q E_{M j} E_{j, \, j + 1} E_{M j} E_{j, \, j + 1} + E_{M j} E^2_{j, \, j + 1} E_{M j}).
\end{equation*}
Now equation (\ref{esee}) implies that (\ref{emmpk}) for all $M < j \le M + N$.

Further, we assume that (\ref{eijs}) is true for some $1 < i < M$ and $M < j \le M + N$, then we have
\begin{multline}
E^2_{i - 1, \, j} = (E_{i - 1, \, i} E_{i j} - q E_{i j} E_{i - 1, \, i})^2 \\
{} = E_{i - 1, \, i} E_{i j} E_{i - 1, \, i} E_{i j} - q E_{i j} E^2_{i - 1, \, i} E_{i j} + q^2 E_{i j} E_{i - 1, \, i} E_{i j} E_{i - 1, \, i}. \label{esmmrmo}
\end{multline}
Here we take into account that $d_i = 1$ for any $i < M$. It follows from the first relation of (\ref{eeei}) that
\begin{equation*}
\ldbr E_{i - 1, \, i}, \, \ldbr E_{i - 1, \, i}, \, E_{i j} \rdbr \rdbr = 0,
\end{equation*}
or, in a more explicit form,
\begin{equation*}
E^2_{i - 1, \, i} E_{i j} - [2]_q \, E_{i - 1, \, i} E_{i j} E_{i - 1, \, i} + E_{i j} E^2_{i - 1, \, i} = 0.
\end{equation*}
Multiplying this equation from the left and from the right by $E_{i j}$, we obtain
\begin{align}
E_{i j} E^2_{i - 1, \, i} E_{i j} - [2]_q \,  E_{i j} E_{i - 1, \, i} E_{i j} E_{i - 1, \, i} & = 0, \\
- [2]_q \, E_{i - 1, \, i} E_{i j} E_{i - 1, \, i} E_{i j} + E_{i j} E^2_{i - 1, \, i} E_{i j} & = 0. \label{eeeee}
\end{align}
It follows that
\begin{equation*}
E_{i j} E_{i - 1, \, i} E_{i j} E_{i - 1, \, i} = E_{i - 1, \, i} E_{i j} E_{i - 1, \, i} E_{i j}.
\end{equation*}
Using this equation in (\ref{esmmrmo}), we come to
\begin{equation*}
E^2_{i - 1, \, j} = -q (- [2]_q \, E_{i - 1, \, i} E_{i j} E_{i - 1, \, i} E_{i j} + E_{i j} E^2_{i - 1, \, i} E_{i j}).
\end{equation*}
Now equation (\ref{eeeee}) gives
\begin{equation*}
E^2_{i - 1, \, j} = 0.
\end{equation*}
Thus, we see that the statement of the proposition is always true.
\end{proof}

%\subsection*{$\calC_{\mathrm{IV}}$}

\begin{proposition} \label{p:4.7}
For any $((i, \, j), \, (m, \, n)) \in \calC_{\mathrm{IV}}$ one has
\begin{align}
& \ldbr E_{i j}, \, E_{m n} \rdbr = E_{i j} E_{m n} - (-1)^{[m] + [j]} E_{m n} E_{i j} = - (q_m - q_m^{-1}) E_{m j} E_{i n}, \label{e37a} \\
& \ldbr F_{m n}, \, F_{i j} \rdbr = F_{m n} F_{i j} - (-1)^{[m] + [j]} F_{i j} F_{m n} = (q_m - q_m^{-1}) F_{i n} F_{m j}. \label{e37b}
\end{align}
\end{proposition}

\begin{proof}
Using proposition \ref{p:4.3}, we get
\begin{multline}
\ldbr E_{i j}, \, E_{m n} \rdbr = (E_{i m} E_{m j} - q_m E_{m j} E_{i m}) E_{m n} - (-1)^{[m] + [j]}  E_{m n} (E_{i m} E_{m j} - q_m E_{m j} E_{i m}) \\
= E_{i m} E_{m j} E_{m n} - (-1)^{[m] + [j]}  E_{m n} E_{i m} E_{m j} \\- q_m (E_{m j} E_{i m} E_{m n} - (-1)^{[m] + [j]} E_{m n} E_{m j} E_{i m}). \label{e37c}
\end{multline}
Proposition \ref{p:4.5} implies
\begin{equation*}
\ldbr E_{m j}, \, E_{m n} \rdbr = E_{m j} E_{m n} - (-1)^{[m] + [j]} q_m^{-1} E_{m n} E_{m j} = 0.
\end{equation*}
Hence, we have
\begin{equation*}
E_{m j} E_{m n} = (-1)^{[m] + [j]} q_m^{-1} E_{m n} E_{m j}, \qquad E_{m n} E_{m j} = (-1)^{[m] + [j]} q_m E_{m j} E_{m n}
\end{equation*}
Using these equations in (\ref{e37c}), we obtain
\begin{multline*}
\ldbr E_{i j}, \, E_{m n} \rdbr = (-1)^{[m] + [j]} q_m^{-1} (E_{i m} E_{m n} - q_m E_{m n} E_{i m}) E_{m j} \\- q_m E_{m j} (E_{i m} E_{m n} - q_m E_{m n} E_{i m}) = (-1)^{[m] + [j]} q_m^{-1} E_{i n} E_{m j} - q_m E_{m j} E_{i n}.
\end{multline*}
Finally, it follows from proposition \ref{p:4.4} that
\begin{equation*}
E_{i n} E_{m j} = (-1)^{[m] + [j]} E_{m j} E_{i n},
\end{equation*}
therefore,
\begin{equation*}
\ldbr E_{i j}, \, E_{m n} \rdbr = - (q_m - q_m^{-1}) E_{m j} E_{i n}.
\end{equation*}
Thus, equation (\ref{e37a}) is true. In the same way one can prove equation (\ref{e37b}).
\end{proof}

%\subsection*{$\calC_{\mathrm{V}} \cup \calC_{\mathrm{VI}}$}

\begin{proposition} \label{p:4.8}
For any $((i, \, j), \, (m, \, n)) \in \calC_{\mathrm{V}}$ one has
\begin{equation*}
\ldbr E_{i j}, \, F_{m n} \rdbr = E_{i j} F_{j n} - F_{j n} E_{i j} = 0, \qquad \ldbr E_{m n}, \, F_{i j} \rdbr = E_{j n} F_{i j} - F_{i j} E_{j n} = 0.
\end{equation*}
For any $((i, \, j), \, (m, \, n)) \in \calC_{\mathrm{VI}}$ one has
\begin{equation*}
\ldbr E_{i j}, \, F_{m n} \rdbr = E_{i j} F_{m n} - F_{m n} E_{i j} = 0, \qquad \ldbr E_{m n}, \, F_{i j} \rdbr = E_{m n} F_{i j} - F_{i j} E_{m n} = 0.
\end{equation*}
\end{proposition}

\begin{proof}
The statement of the proposition is a direct consequence of the defining relation (\ref{glrc}).
\end{proof}

%\subsection*{$\calC_{\mathrm{II}}$}

\begin{proposition} \label{p:4.9}
For any $((i, \, j), \, (m, \, n)) \in \calC_{\mathrm{II}}$ one has
\begin{align}
& \ldbr E_{i j}, \, F_{m n} \rdbr = E_{i j} F_{m n} - (-1)^{[m] + [n]} F_{m n} E_{i j} = 0,  \label{e39a} \\ 
& \ldbr E_{m n}, \, F_{i j} \rdbr = E_{m n} F_{i j} - (-1)^{[m] + [n]} F_{i j} E_{m n} = 0. \label{e39b}
\end{align}
\end{proposition}

\begin{proof}
Let $1 < k \le M + N - 2$. Prove equation (\ref{e39a}) for $i = k - 1$, $m = k$, $n = k + 1$ and $j = k + 2$. We have
\begin{equation*}
E_{k - 1, \, k + 2} = E_{k -1} E_k E_{k + 1} - q_{k + 1} E_{k - 1} E_{k + 1} E_k - q_k E_{k + 1} E_{k + 1} E_{k - 1} + q_k q_{k + 1} E_{k + 1} E_k E_{k - 1}.
\end{equation*}
It follows that
\begin{multline*}
\ldbr E_{i j}, \, F_{m n} \rdbr = E_{k - 1}\ldbr E_k, \, F_k \rdbr E_{k + 1} - q_{k + 1} E_{k - 1} E_{k + 1} \ldbr E_k, \, F_k \rdbr \\*
- q_k \ldbr E_k, \, F_k \rdbr E_{k + 1} E_{k - 1} + q_k q_{k + 1} E_{k + 1} \ldbr E_k, \, F_k \rdbr E_{k - 1}.
\end{multline*}
Using the defining relation (\ref{glrc}) and proposition \ref{p:4.1}, we obtain
\begin{equation*}
\ldbr E_{k - 1, \, k + 2}, \, F_{k, \, k + 1} \rdbr = 0.
\end{equation*}
Now, let $1 < i < k < j - 1 \le M + N - 1$ and
\begin{equation}
\ldbr E_{i j}, \, F_{k, \, k + 1} \rdbr = E_{i j} F_{k, \, k + 1} - (-1)^{[k] + [k + 1]} F_{k, \, k + 1} E_{i j} = 0.  \label{e39c}
\end{equation}
We have
\begin{multline*}
\ldbr E_{i - 1, \, j}, \, F_{k, \, k + 1} \rdbr = \ldbr E_{i - 1, \, i} E_{i j} - q_i E_{i j} E_{i - 1, \, i} , \, F_{k, \, k + 1} \rdbr = E_{i - 1, i} E_{i j} F_{k, \, k + 1} \\- (-1)^{[k] + [k + 1]} F_{k, \, k + 1} E_{i - 1, \, i} E_{i j} - q_j (E_{i j} E_{i - 1, \, i}F_{k, \, k + 1} - (-1)^{[k] + [k + 1]} F_{k, \, k + 1} E_{i j} E_{i - 1, \, i}).
\end{multline*}
It follows from proposition \ref{p:4.8} that
\begin{equation*}
\ldbr E_{i - 1, \, j}, \, F_{k, \, k + 1} \rdbr = 0.
\end{equation*}
Further, let $1 \le i < k < j - 1 \le M + N - 2$ and equation (\ref{e39c}) is true. We obtain
\begin{multline*}
\ldbr E_{i, \, j + 1}, \, F_{k, \, k + 1} \rdbr = \ldbr E_{i j} E_{j, \, j + 1} - q_j E_{j, \, j + 1} E_{i j} , \, F_{k, \, k + 1} \rdbr = E_{i j} E_{j, \, j + 1} F_{k, \, k + 1} \\- (-1)^{[k] + [k + 1]} F_{k, \, k + 1} E_{i j} E_{j, \, j + 1} - (E_{j, j + 1} E_{i j} F_{k, \, k + 1} - (-1)^{[k] + [k + 1]} q_j F_{k, \, k + 1} E_{j, \, j + 1} E_{i j}),
\end{multline*}
and proposition \ref{p:4.8} implies that
\begin{equation*}
\ldbr E_{i, \, j + 1}, \, F_{k, \, k + 1} \rdbr = 0.
\end{equation*}
Hence, we have
\begin{equation*}
\ldbr E_{i j}, \, F_{k, \, k + 1} \rdbr = 0
\end{equation*}
for all possible $i$, $j$ and $k$. Assume now that for some $1 \le i < m < n < M + N - 1$ we have
\begin{equation*}
\ldbr E_{i j}, \, F_{m n} \rdbr = E_{i j} E_{m n} - (-1)^{[m] + [n]} E_{m n} E_{i j} = 0.
\end{equation*}
Then, we obtain
\begin{multline*}
\ldbr E_{i j}, \, F_{m, \, n + 1} \rdbr = \ldbr E_{i j}, \, F_{n, \, n + 1} F_{m n} - q^{-1}_n F_{m n} F_{n, \, n + 1} \rdbr = E_{i j} F_{n, \, n + 1} F_{m n} \\*
- (-1)^{[m] + [n + 1]} F_{n, \, n + 1} F_{n m} E_{i j} - q_n^{-1} (E_{i j} F_{m n} F_{n, \, n + 1} - (-1)^{[m] + [n + 1]} F_{m n} F_{n, \, n + 1} E_{i j}) = 0.
\end{multline*}
Thus, equation (\ref{e39a}) is true. In the same way one can prove equation (\ref{e39b}).
\end{proof}

%\subsection*{$\calC_{\mathrm{I}} \cup \calC_{\mathrm{III}}$}

\begin{proposition} \label{p:4.10}
For any $((i, \, j), \, (m, \, n)) \in \calC_{\mathrm{I}}$ one has
\begin{align}
& \ldbr E_{i j}, \, F_{i n} \rdbr = E_{i j} F_{i n} - (-1)^{[i] + [j]} F_{i n} E_{i j} = - (-1)^{[i] + [j]} q^{- d_i K_i + d_j K_j} F_{j n}^{}, \label{e310a} \\
& \ldbr E_{i n}, \, F_{i j} \rdbr = E_{i n} F_{i j} - (-1)^{[i] + [j]} F_{i j} E_{i n} = - (-1)^{[i] + [j]} E_{j n} \, q^{d_i K_i - d_j K_j}. \label{e310b}
\end{align}
For any $((i, \, j), \, (m, \, n)) \in \calC_{\mathrm{III}}$ one has
\begin{align}
& \ldbr E_{i j}, \, F_{m n} \rdbr = E_{i j} F_{m j} - (-1)^{[m] + [j]} F_{m j} E_{i j} = q^{ - d_m K_m + d_j K_j} E_{i m}, \label{e310c} \\
& \ldbr E_{m n}, \, F_{i j} \rdbr = E_{m j} F_{i j} - (-1)^{[m] + [j]} F_{i j} E_{m j} = F_{i m} q^{d_m K_m - d_j K_j}. \label{e310d}
\end{align}
\end{proposition}

\begin{proof}
We first prove equation (\ref{e310a}) for $j = i + 1$. Using proposition \ref{p:4.3}, we obtain
\begin{multline*}
\ldbr E_{i, \, i + 1}, \, F_{i n} \rdbr = \ldbr E_{i, \, i + 1}, \, F_{i + 1, \, n} F_{i, \, i + 1} - q_{i + 1}^{-1} F_{i, \, i + 1} F_{i + 1, \, n} \rdbr \\
= E_{i, \, i + 1} F_{i + 1, \, n} F_{i, \, i + 1} - (-1)^{[i] + [i + 1]} F_{i + 1, \, n} F_{i, \, i + 1} E_{i, \, i + 1} \\- q_{i + 1}^{-1} (E_{i, \, i + 1} F_{i, \, i + 1} F_{i + 1, \, n} - (-1)^{[i] + [i + 1]} F_{i, \, i + 1} F_{i + 1, \, n} E_{i, \, i + 1}).
\end{multline*}
Further, proposition \ref{p:4.8} gives
\begin{multline*}
\ldbr E_{i, \, i + 1}, \, F_{i n} \rdbr = F_{i + 1, \, n} \ldbr E_i, \, F_i \rdbr - q_{i + 1}^{-1} \ldbr E_i, F_i \rdbr F_{i + 1, \, n} \\*
= (q_i - q_i^{-1})^{-1} (F_{i + 1, \, n} (q^{d_i K_i - d_{i + 1} K_{i + 1}} - q^{- d_i K_i + d_{i + 1} K_{i + 1}}) \\ 
- q_{i + 1}^{-1} (q^{d_i K_i - d_{i + 1} K_{i + 1}} - q^{- d_i K_i + d_{i + 1} K_{i + 1}}) F_{i + 1, \, n}).
\end{multline*}
and, using proposition \ref{p:4.1}, we come to
\begin{equation*}
\ldbr E_{i, \, i + 1}, \, F_{i n} \rdbr = - (q_{i + 1} - q_{i + 1}^{-1})(q_i - q_i^{-1})^{-1} q^{- d_i K_i + d_{i + 1} K_{i + 1}} F_{i + 1, \, n} .
\end{equation*}
Finally, it follows from (\ref{kappaqi}) that
\begin{equation*}
\ldbr E_{i, \, i + 1}, \, F_{i n} \rdbr = - (-1)^{- [i] + [i + 1]}  q^{- d_i K_i + d_{i + 1} K_{i + 1}} F_{i + 1, \, n}.
\end{equation*}
Now, let $1 \le i < j < n - 1 \le M + N -1$ and equation
\begin{equation}
\ldbr E_{i j}, \, F_{i n} \rdbr = - (-1)^{[i] + [j]} q^{- d_i K_i + d_j K_j} F_{j n} \label{e310e}
\end{equation}
be true. Using proposition \ref{p:4.3}, we obtain
\begin{multline*}
\ldbr E_{i, \, j + 1}, \, F_{i n} \rdbr = \ldbr E_{i j} E_{j, \, j + 1} - q_j E_{j, \, j + 1} E_{i j}, \, F_{i n} \rdbr \\
= E_{i j} E_{j, \, j + 1} F_{i n} - (-1)^{[i] + [j + 1]} F_{i n} E_{i j} E_{j, \, j + 1} \\
- q_j(E_{j, \, j + 1} E_{i j} F_{i n} - (-1)^{[i] + [j + 1]} F_{i n} E_{j, \, j + 1} E_{i j}). 
\end{multline*}
It follows from proposition \ref{p:4.4}, equation (\ref{e310e}) and proposition \ref{p:4.1} that
\begin{multline*}
\ldbr E_{i, \, j + 1}, \, F_{i n} \rdbr = (-1)^{[j] + [j + 1]} \ldbr E_{i j}, \, F_{i n} \rdbr E_{j, \, j + 1} - q_j E_{j, \, j + 1} \ldbr E_{i j}, \, F_{i n} \rdbr \\
= (-1)^{[i] + [j]} q^{-d_i K_i + d_j K_j} \ldbr E_{j, \, j + 1}, \, F_{j n} \rdbr = - (-1)^{[i] + [j + 1]} q^{- d_i K_i + d_{j + 1} K_{j + 1}} F_{j + 1, \, n}.
\end{multline*}
We see that equation (\ref{e310a}) is always true. In the same way one can prove equations (\ref{e310b}),  (\ref{e310c}) and  (\ref{e310d}).
\end{proof}

%\subsection*{$\calC$}

\begin{proposition} \label{p:4.11}
For any $1 \le i < j \le M + N$ we have
\begin{equation*}
\ldbr E_{i j}, \, F_{i j} \rdbr = E_{i j} F_{i j} - (-1)^{[i] + [j]} F_{i j} E_{i j} = \frac{q^{d_i K_i - d_j K_j} - q^{- d_i K_i + d_j K_j}}{q_i^{} - q_i^{-1}}.
\end{equation*}
\end{proposition}

\begin{proof}
The statement of the proposition is certainly true for $j = i + 1$. Let us consider the case when $j - i > 1$. It follows from proposition \ref{p:4.3} that
\begin{multline}
\ldbr E_{i j}, \, F_{i j} \rdbr = \ldbr E_{i j}, \, F_{i + 1, \, j} F_{i, \, i + 1} - q^{-1}_{i + 1} F_{i, \, i + 1} F_{i + 1, \, j} \rdbr \\
= E_{i j} F_{i + 1, \, j} F_{i, \, i + 1} - (-1)^{[i] + [j]} F_{i + 1, \, j} F_{i, \, i + 1} E_{i j} \\
- q_{i + 1}^{-1} (E_{i j} F_{i, \, i + 1} F_{i + 1, \, j} - (-1)^{[i] + [j]} F_{i, \, i + 1} F_{i + 1, \, j} E_{i j}). \label{e312a}
\end{multline}
Equation (\ref{e310b}) implies
\begin{align*}
& F_{i, \, i + 1} E_{i j} = (-1)^{[i] + [i + 1]} E_{i j} F_{i, \, i + 1} + E_{i + 1, \, j} q^{d_i K_i - d_{i + 1} K_{i + 1}}, \\
& E_{i j} F_{i, \, i + 1} = (-1)^{[i] + [i + 1]} F_{i, \, i + 1} E_{i j} - (-1)^{[i] + [i + 1]} E_{i + 1, \, j} q^{d_i K_i - d_{i + 1} K_{i + 1}}.
\end{align*}
Using these equations in (\ref{e312a}), we obtain
\begin{multline}
\ldbr E_{i j}, \, F_{i j} \rdbr = \ldbr E_{i j}, \, F_{i + 1, \, j} \rdbr F_{i, \, i + 1} +  q_{i + 1}^{-1} (-1)^{[i] + [ i + 1]} F_{i, \, i + 1} \ldbr E_{i j}, \, F_{i + 1, \, j} \rdbr \\
+ q_{i + 1}^{-1} (-1)^{[i] + [i + 1]} E_{i + 1, \, j} q^{d_i K_i - d_{i + 1} K_{i + 1}} F_{i + 1, \, j} \\ - (-1)^{[i] + [j]} F_{i + 1, \, j} E_{i + 1, j} q^{d_i K_i - d_{i + 1} K_{i + 1}}.
\end{multline}
We have
\begin{equation*}
\ldbr E_{i j}, \, F_{i + 1, \, j} \rdbr = q^{- d_{i + 1} K_{i + 1} + d_j K_j} E_{i, \, i + 1},
\end{equation*}
see proposition \ref{p:4.11}. Taking this equation, and proposition \ref{p:4.1} and equation (\ref{kappaqi}) into account, we come to the equation
\begin{multline*}
\ldbr E_{i j}, \, F_{i j} \rdbr = q^{-d_{i + 1} K_{i + 1} + d_j K_j} \ldbr E_{i, \, i + 1}, \, F_{i, \, i + 1} \rdbr \\+ (-1)^{[i] + [i + 1]} \ldbr E_{i + 1, \, j}, \, F_{i + 1, \, j} \rdbr q^{d_i K_i - d_{i + 1} K_{i + 1}} = \frac{q^{d_i K_i - d_j K_j} - q^{- d_i K_i + d_j K_j}}{q_i^{} - q^{-1}_i}.
\end{multline*}
That was to be proved.
\end{proof}

%\subsection*{$\calC_{\mathrm{IV}}$}

\begin{proposition} \label{p:4.12}
For any $((i, \, j), \, (m, \, n)) \in \calC_{\mathrm{IV}}$ one has
\begin{align*}
& \ldbr E_{i j}, \, F_{m n} \rdbr = E_{i j} F_{m n} - (-1)^{[m] + [j]} F_{m n} E_{i j} = -(q_j^{} - q_j^{-1}) q^{-d_m K_m + d_j K_j} F_{j n}^{} E_{i m}^{}, \\
& \ldbr E_{m n}, \, F_{i j} \rdbr = E_{m n} F_{i j} - (-1)^{[m] + [j]} F_{i j} E_{m n} = (q_j^{} - q_j^{-1})F_{i m}^{} E_{j n}^{}  q^{d_m K_m - d_j K_j} .
\end{align*}
\end{proposition}

\begin{proof} It follows from propositions \ref{p:4.3} and \ref{p:4.9} that
\begin{multline*}
\ldbr E_{i j}, \, F_{m n} \rdbr = \ldbr E_{i m} E_{m j} - q_m E_{m j} E_{i m}, \, F_{m n} \rdbr 
\\= E_{i m} E_{m j} F_{m n} - (-1)^{[m] + [j]} F_{m n} E_{i m} E_{m j} - q_m (E_{m j} E_{i m} F_{m n} - (-1)^{[m] + [j]} F_{m n} E_{m j} E_{i m}) \\= E_{i m} \ldbr E_{m j}, \, F_{m n} \rdbr - q_m \ldbr E_{m j}, \, F_{m n} \rdbr E_{i m}.
\end{multline*}
Now, using equation (\ref{e310a}), proposition \ref{p:4.1} and proposition \ref{p:4.8}, we get
\begin{equation*}
\ldbr E_{i j}, \, F_{m n} \rdbr = (-1)^{[m] + [j]} (q_m^{} - q_m^{-1}) q^{- d_m K_m + d_j K_j} F_{j n} E_{i m}.
\end{equation*}
Now, taking into account equations (\ref{kappaqi}), we see that the first equation of the proposition is true. The second equation can be proved similarly.
\end{proof}

One can get convinced that the propositions \ref{p:4.1}--\ref{p:4.12} allow us to reduce any monomial on the Poincar\'e--Birkhoff--Witt generators to the ordered form (\ref{41c}).

\section{Conclusions}

We have derived the commutation relations for the Poincar\'e--Birkhoff--Witt generators of the quantum algebra $\uqglmn$. Our results do not fully coincide with the results of the papers \cite{Zha93, Tsu14, Tsu19}. We are planning to use the obtained relations for constructing of $q$-oscillator representations of the positive Borel subalgebra of the quantum superalgebra $\uqglmn$.

\subsection*{Acknowledgments}

This work was supported in part by the RFBR grant \#~20-51-12005.

\bibliographystyle{amsrusunsrt}
\bibliography{IntegrableSystems}

\end{document}